\documentclass[runningheads]{llncs}
\usepackage{graphicx}
\usepackage[T1]{fontenc}

\usepackage{amsmath,amsfonts,amssymb,enumitem,bbm,xcolor,mathtools}
\usepackage{etoolbox}
\usepackage{tikz}

\usepackage{hyperref}
\usepackage{cleveref}

\newcommand{\Define}    [1] {\textbf{#1}}

\newcommand{\zero}      { \mathbf{0} }
\newcommand{\one}       { \mathbf{1} }

\newcommand{\bigmeet}   { \bigwedge }
\newcommand{\bigjoin}   { \bigvee }

\newcommand{\Network}       { \mathsf{F} }

\newcommand{\Kernel}        { \mathsf{K} }

\newcommand{\IndependentSets}   { \mathrm{I} }
\newcommand{\Kernels}           { \mathrm{K} }

\newcommand{\Fix}{\mathrm{Fix}}

\newcommand{\Neighbourhood}     { N }

\begin{document}

\title{Words fixing the kernel network and maximum independent sets in graphs}

\author{Maximilien Gadouleau\inst{1}\orcidID{0000-0003-4701-738X} and David C. Kutner\inst{1}\orcidID{0000-0003-2979-4513} }
\authorrunning{M. Gadouleau and D. C. Kutner}

\institute{Durham University, UK\\
\email{\{m.r.gadouleau,david.c.kutner\}@durham.ac.uk}
}
\maketitle              
\begin{abstract}
The simple greedy algorithm to find a maximal independent set of a graph can be viewed as a sequential update of a Boolean network, where the update function at each vertex is the conjunction of all the negated variables in its neighbourhood. In general, the convergence of the so-called kernel network is complex. A word (sequence of vertices) fixes the kernel network if applying the updates sequentially according to that word. We prove that determining whether a word fixes the kernel network is coNP-complete. We also consider the so-called permis, which are permutation words that fix the kernel network. We exhibit large classes of graphs that have a permis, but we also construct many graphs without a permis.

\end{abstract}



\section{Introduction} \label{section:introduction}

\subsection{Background}


A simple greedy algorithm to find a maximal independent set in a graph works as follows. Starting with the empty set, visit each vertex in the graph, and add it to the set whenever none of its neighbours are already in the set. This can be interpreted in terms of Boolean networks as follows: starting with the all-zero configuration $x$, update one vertex $v$ at a time according to the update function $\bigmeet_{u \sim v} \neg x_u$. For the final configuration $y$, the set of ones is a maximal independent set, regardless of the order in which the vertices have been updated.

We refer to the Boolean network where the update function is the conjunction of all the negated variables in the neighbourhood of a vertex as the \emph{kernel network} on the graph. We use this terminology, as kernels are the natural generalisation of maximal independent sets to digraphs (they are the independent dominating sets). This class of networks has been the subject of some study; we refer to two works in particular.

    In \cite{ARS17}, the fixed points of different conjunctive networks on graphs are studied. In particular, it is shown that the set of fixed points of the kernel network is the set of (configurations whose coordinates equal to one are) maximal independent sets of the graph. They further prove that for square-free graphs, the kernel network is the conjunctive network that maximises the number of fixed points.

    In a completely different setting, Yablo discovered the first non-self-referential paradox in \cite{Yab93}. This paradox is based on the fact that the kernel network on a transitive tournament on $\mathbb{N}$ has no fixed point. The study of acyclic digraphs that admit a paradox is continued further in \cite{RRM13}, where the kernel network is referred to as an $\mathcal{F}$-system.

Boolean networks are used to model networks of interacting entities. As such, it is natural to consider a scenario whereby the different entities update their state at different times. This gives rise to the notion of sequential (or asynchronous) updates. The problem of whether a Boolean network converges sequentially goes back to the seminal result by Robert on acyclic interaction graphs \cite{Rob80}; further results include  \cite{Gol85,GM12,NS17}.
Recently, \cite{AGRS20} introduced the concept of a fixing word: a sequence of vertices such that updating vertices according to that sequence will always lead to a fixed point, regardless of the initial configuration. Large classes of Boolean networks have short fixing words \cite{AGRS20,GR18}.

\subsection{Contributions and outline}



    
Our main result is to show that determining whether a word fixes the kernel network is an NP-hard problem. Our seminal remark is that if $w$ is any permutation of the vertices, then $w$ maps any configuration to an independent set (we say $w$ prefixes the kernel network) and $w$ maps any independent set to a kernel (we say $w$ suffixes the kernel network), and as such $ww$ fixes the kernel network. Once again, whether a word prefixes or suffixes the kernel network only depends on the set of vertices visited by the word. Determining whether a word prefixes the kernel network can be done in polynomial time, while it is coNP-complete to determine whether it suffixes the kernel network. We then determine the sets of vertices $S$ for which there exists a word fixing the kernel network that only visits $S$; deciding whether $S$ is one such set is NP-hard. We use the intractability of that last problem to prove our main result.
    
We then go back to our interpretation of the kernel network in terms of the greedy algorithm for finding a maximal independent set. In that algorithm, the initial configuration is fixed and the permutation of vertices is arbitrary. We then consider fixing the permutation and varying the initial configuration instead. Thus we introduce the notion of a permis, i.e. a permutation that fixes the kernel network. We exhibit large classes of graphs which do have a permis, and some examples and constructions of graphs which do not have a permis.

The rest of the paper is organised as follows. Some necessary background on Boolean networks is reviewed in Section \ref{section:preliminaries}. In Section \ref{subsection:prefix_suffix} we classify those words which prefix or suffix the kernel network and show that determining whether a word fixes the kernel network is coNP-complete. Lastly, we study graphs that have a permis in Section \ref{subsection:permis}.

Due to space limitations, some proofs are given in the appendix, and their sketches are given in the main text instead. We presume that the reader is familiar with some basic graph theory; otherwise, they are directed to \cite{BM08}.

\section{Preliminaries} \label{section:preliminaries}

\subsection{Boolean networks} \label{subsection:background_BN}

A \Define{configuration} on a graph $G = (V,E)$ is $x \in \{0,1\}^V = (x_v : v \in V)$, where $x_v \in \{0,1\}$ is the state of the vertex $v$ for all $v$. We denote $\one( x ) = \{ v \in V: x_v = 1 \}$ and $\zero( x ) = \{ v \in V: x_v = 0 \}$. For any set of vertices $S \subseteq V$, we denote $x_S = (x_v: v \in S)$. We denote the all-zero (all-one, respectively) configuration by $0$ (by $1$, respectively), regardless of its length. 

A \Define{Boolean network} is a mapping $\Network : \{0,1\}^V \to \{0,1\}^V$. 
For any Boolean network $\Network$ and any $v \in V$, the update of the state of vertex $v$ is represented by the network $\Network^v: \{0,1\}^V \to \{0,1\}^V$ where $\Network^v(x)_v = \Network( x )_v$ and $\Network^v( x )_u = x_u$ for all other vertices $u$.
We extend this notation to words as follows: if $w = w_1 \dots w_l$ then
\[
	\Network^w = \Network^{ w_l } \circ \dots \circ \Network^{ w_2 } \circ \Network^{ w_1 }.
\]
Unless otherwise specified, we let $x$ be the initial configuration, $w = w_1 \dots w_l$ a word, $y = \Network^w( x )$ be the final configuration, and for all $0 \le a \le l$, $y^a = \Network^{w_1 \dots w_a}( x )$ be an intermediate configuration, so that $x = y^0$ and $y = y^l$.

The set of \Define{fixed points} of $\Network$ is $\Fix( \Network ) = \{ x \in \{0,1\}^V : \Network( x ) = x \}$. 
The word $w$ \Define{fixes} $\Network$ if for all $x$, $\Network^w( x ) \in \Fix( \Network )$.

\subsection{The kernel network} \label{subsection:networks}


    
    

Let $G = (V,E)$ be a graph. The \Define{kernel network} on $G$, denoted as $\Kernel(G)$ is defined by
    \[
        \Kernel(x)_v = \bigmeet_{u \sim v} \neg x_u,
    \]
with $\Kernel(x)_v = 1$ if $\Neighbourhood(v) = \emptyset$.

An \Define{independent set} $I$ is a set such that $ij \notin E$ for all $i,j \in I$. 
The collection of all configurations $x$ of $G$ such that $\one( x )$ is an independent set of $G$ is denoted by $\IndependentSets(G)$. 
A \Define{dominating set} $D$ is a set such that for every vertex $v \in V$, either $v \in D$ or there exists $u \in D$ such that $uv \in E$. A \Define{kernel} $K$ is a dominating independent set. Equivalently, a kernel is a \Define{maximal independent set} of $G$, i.e. an independent set $K$ such that there is no independent set $J \supset K$. The collection of all configurations $x$ of $G$ such that $\one( x )$ is a kernel of $G$ is denoted by $\Kernels(G)$.  It is easily seen (for instance, in \cite{ARS17}) that $\Fix( \Kernel( G ) ) = \Kernels( G )$.

\section{Words fixing the kernel network} \label{subsection:prefix_suffix}

We now focus on words fixing the kernel network. Whether a word fixes the kernel network does not only depend on the set of vertices it visits. For example, if $G$ is the path on the three vertices $a, b, c$ with edges $ab, bc$, then $w = abc$ does not fix $\Kernel( G )$ (if $x = 111$, then $y = 001$), while it is easily checked that $w = acb$ does fix $\Kernel(G)$. In general, characterising the fixing words for $\Kernel(G)$ remains an open problem. However, we manage to prove that deciding whether a word fixes the kernel network is computationally hard.

We define \textsc{Fixing Word} to be the decision problem asking, for an instance $(G, w)$, whether $w$ fixes $\Kernel( G )$.

\begin{theorem} \label{theorem:complexity_fixing_word_K}
\textsc{Fixing Word} is coNP-complete.
\end{theorem}



The rest of this section is devoted to the proof of Theorem \ref{theorem:complexity_fixing_word_K}.

The set of vertices visited by a word $w$ is denoted by $[w] = \{v \in V : \exists a, v = w_a \}$. A \Define{permutation} of $V$ (or of $G$) is a word $w = w_1 \dots w_n$ such that $[w] = V$ and $w_a \ne w_b$ for all $a \ne b$. If $w$ is a permutation of $G$, then $ww$ fixes $\Kernel( G )$: for any initial configuration $x$, $\Kernel^w( x ) \in \IndependentSets( G )$; then for any $y \in \IndependentSets( G )$, $\Kernel^w( y ) \in \Kernels( G )$. 

Accordingly, we say that $w^\mathrm{p}$ \Define{prefixes} $\Kernel( G )$ if $\Kernel^{ w^\mathrm{p} }(x) \in \IndependentSets( G )$ for all $x \in \{0,1\}^n$, and that $w^\mathrm{s}$ \Define{suffixes} $\Kernel( G )$ if $\Kernel^{ w^\mathrm{s} }(y) \in \Kernels( G )$ for all $y \in \IndependentSets( G )$. In that case, for any word $\omega$, $w^\mathrm{p}\omega$ also prefixes $\Kernel( G )$ and $\omega w^\mathrm{s}$ also suffixes $\Kernel( G )$. Clearly, if $w = w^\mathrm{p} w^\mathrm{s}$, where $w^\mathrm{p}$ prefixes $\Kernel(G)$ and $w^\mathrm{s}$ suffixes $\Kernel(G)$, then $w$ fixes $\Kernel(G)$. We can be more general, as shown below.

\begin{proposition} \label{proposition_prefix_and_suffix}
If $w = w_1 \dots w_l$ where $w_1 \dots w_a$ prefixes $\Kernel(G)$, $w_b \dots w_l$ suffixes $\Kernel( G )$, and $[w_b \dots w_a]$ is an independent set of $G$ for some $0 \le a, b \le l$, then $w$ fixes $\Kernel( G )$.
\end{proposition}

\begin{proof}
First, suppose $a < b$, so that $w = w_1 \dots w_a \dots w_b \dots w_l$. As mentioned above, $w^\mathrm{p} = w_1 \dots w_{b-1}$ prefixes $\Kernel(G)$ and $w^\mathrm{s} = w_b \dots w_l$ suffixes $\Kernel(G)$, hence $w = w^\mathrm{p} w^\mathrm{s}$ fixes $\Kernel(G)$.

Second, suppose $a \ge b$, so that $w = w_1 \dots w_b \dots w_a \dots w_l$. It is easily seen that if $u \not\sim v$, $\Kernel^{vv} = \Kernel^{ v }$ and $\Kernel^{uv} = \Kernel^{vu}$. As such, 
\[
    \Kernel^w = \Kernel^{w_1 \dots w_b w_b \dots w_a w_a \dots w_l} = \Kernel^{w_1 \dots w_b \dots w_a w_b \dots w_a \dots w_l},
\]
and again if we let $w^\mathrm{p} = w_1 \dots w_a$ and $w^\mathrm{s} = w_b \dots w_l$, we have $\Kernel^w = \Kernel^{ w^\mathrm{p} w^\mathrm{s} }$, hence $w$ fixes $\Kernel( G )$.
\qed \end{proof}





We now characterise the words that prefix (or suffix) the kernel network. Interestingly, those properties depend only on $[w]$.

\begin{proposition} 
\label{proposition:prefix}
Let $G$ be a graph. Then $w$ prefixes $\Kernel(G)$ if and only if $[w]$ is a vertex cover of $G$.
\end{proposition}

\begin{proof}
Suppose $[w]$ is a vertex cover of $G$ and that $y = \Kernel^w( x ) \notin \IndependentSets( G )$, i.e. $y_{uv} = 11$ for some edge $uv$ of $G$. Without loss, let the last update in $\{u,v\}$ be $v$, i.e. there exists $a$ such that $w_a = v$ and $w_b \notin \{ u, v \}$ for all $b > a$. Let $z = \Kernel^{ w_1 \dots w_{a-1} }( x )$, then $z_u = y_u = 1$ hence $y_v = 0$, which is the desired contradiction.

Conversely, if $[w]$ is not a vertex cover, then there is an edge $uv \in E$ such that $[w] \cap \{u,v\} = \emptyset$. Therefore, for any $x$ with $x_{uv} = 11$, we have $y_{uv} = 11$ as well.
\qed \end{proof}

\begin{corollary} \label{corollary:complexity_prefix}
Given a graph $G$ and a word $w$, determining whether $w$ prefixes $\Kernel(G)$ is in P.
\end{corollary}




A subset $S$ of vertices of a graph is a \Define{colony} if there exists an independent set $I$ such that $S \subseteq \Neighbourhood(I)$. Alternatively, a colony is a set $S$ such that $V \setminus S$ contains a maximal independent set. A subset $W$ of vertices is a \Define{dominion} if there exists $v \in V \setminus W$ such that $W \cap \Neighbourhood(v)$ is a colony of $G - v$. A \Define{non-dominion} is a set of vertices that is not a dominion.

\begin{proposition}
\label{proposition:suffix}
Let $G$ be a graph. Then the word $w$ suffixes $\Kernel(G)$ if and only if $[w]$ is a non-dominion of $G$.
\end{proposition}

\begin{proof}
Suppose $[w]$ is a dominion of $G$, i.e. there exists an independent set $I$ and a vertex $v \notin [w]$ such that $W = [w] \cap \Neighbourhood(v)$ is in the neighbourhood of $I$. Let $x$ such that $x_I = 1$ and $x_{V \setminus I} = 0$, and let $y = \Kernel^w(x)$. Then for any $u \in W$, $u$ has a neighbour in $I$, hence $y_u = 0$; thus $y_{\Neighbourhood[v]} = 0$ and $w$ does not suffix $\Kernel$.

Conversely, suppose there exists $x$ and $v$ such that $y = \Kernel^w(x)$ with $y_{\Neighbourhood[v]} = 0$. Then $x_{\Neighbourhood[v]} = 0$. Let $W = [w] \cap \Neighbourhood(v)$ and $I = \one( y ) \cap \Neighbourhood(W)$; we note that $I$ is an independent set. For each $u \in W$, we have $y_u = 0$ hence there exists $i \in I$ such that $u \in \Neighbourhood(i)$. Therefore, $W \subseteq \Neighbourhood(I)$ and $W$ is a colony of $G - v$.
\qed \end{proof}

The \textsc{Colony} (respectively, \textsc{Dominion}, \textsc{Non-Dominion}) problem asks, given a graph $G$ and set $T$, if $T$ a colony (resp. a dominion, a non-dominion) of $G$.

\begin{theorem} \label{theorem:complexity_suffix}
Given $G$ and $w$, determining whether $w$ suffixes $\Kernel(G)$ is coNP-complete.
\end{theorem}

\begin{proof}[Sketch]
The proof is by successive reductions: \textsc{Set Cover} to \textsc{Colony} to \textsc{Dominion}. The full proof is given in Appendix \ref{appendix:proof_of_theorem:complexity_suffix}.
\end{proof}

We now characterise the sets of vertices $S$ visited by fixing words of the kernel network. Interestingly, those are the same sets $S$ such that $ww$ is a fixing word for any permutation $w$ of $S$.

\begin{proposition} \label{proposition:fixing_sets}
Let $S$ be a subset of vertices of $G$. The following are equivalent.
\begin{enumerate}
    \item \label{item:word_fixing}
    There exists a word $w$ with $[w] = S$ that fixes $\Kernel( G )$.

    \item \label{item:words_prefix_suffix}
    For all $w^\mathrm{p}$, $w^\mathrm{s}$ such that $[ w^\mathrm{p} ] = [ w^\mathrm{s} ] = S$, the word $w^\mathrm{p} w^\mathrm{s}$ fixes $\Kernel( G )$.

    \item \label{item:subset}
    $S$ is a vertex cover and a non-dominion.
\end{enumerate}
\end{proposition}

\begin{proof}
Clearly, $\ref{item:words_prefix_suffix} \implies \ref{item:word_fixing}$. We now prove $\ref{item:word_fixing} \implies \ref{item:subset}$. Since $w$ prefixes $\Kernel( G )$, $S = [w]$ is a vertex cover by Proposition \ref{proposition:prefix}; similarly, since $w$ suffixes $\Kernel( G )$, $S = [w]$ is a non-dominion by Proposition \ref{proposition:suffix}. Finally, we prove $\ref{item:subset} \implies \ref{item:words_prefix_suffix}$. Since $S$ is a vertex cover, then by Proposition \ref{proposition:prefix} $w^\mathrm{p}$ prefixes $\Kernel( G )$; similarly, by Proposition \ref{proposition:suffix} $w^\mathrm{s}$ suffixes $\Kernel( G )$. Therefore, $w^\mathrm{p} w^\mathrm{s}$ fixes $\Kernel( G )$.
\qed \end{proof}

Let \textsc{Fixing Set} be the decision problem, where the instance is $(G, S)$ and the question is: does there exist a word $w$ with $[w] = S$ that fixes $\Kernel( G )$, or equivalently is $S$ a vertex cover and a non-dominion?

\begin{theorem} \label{theorem:complexity_fixing_set}
\textsc{Fixing Set} is NP-hard.
\end{theorem}

\begin{proof}
The proof is by reduction \textsc{Non-Dominion} (which is coNP-complete) to \textsc{Fixing Set}. The full proof is given in Appendix \ref{appendix:proof_of_theorem:complexity_fixing_set}. 
\end{proof}

We now finalise the proof of Theorem \ref{theorem:complexity_fixing_word_K}.


\begin{proof}[of Theorem \ref{theorem:complexity_fixing_word_K}]
\textsc{Fixing Word} is in coNP; the certificate being a configuration $x$ such that $\Kernel^w( x ) \notin \Kernels( G )$. The proof of hardness is by reduction from \textsc{Fixing Set}, which is NP-hard, as shown in Theorem \ref{theorem:complexity_fixing_set}. Let $(G, S)$ be an instance of \textsc{Fixing Set}, then consider the instance $(G, w = \omega \omega)$ of \textsc{Fixing Word}, where $\omega$ is a permutation of $S$. Then Proposition \ref{proposition:fixing_sets} shows that $w$ fixes $\Kernel( G )$ if and only if $S$ is a vertex cover and a non-dominion.
\end{proof}

The complexity of determining the length of a shortest fixing word for the kernel network remains open.

\begin{question} \label{question:complexity_shortest_fixing_word_K}
What is the complexity of the following optimisation problem: given $G$, what is the length of a shortest fixing word for $\Kernel(G)$?
\end{question}


\section{Graphs with a permis} \label{subsection:permis}

Let $G = (V,E)$ be a graph. The greedy algorithm to find a maximal independent set of $G$ fixes the initial configuration $x$ to $0$, and varies the permutation $w$, while always obtaining a maximal independent set $y \in \Kernels( G )$. We now turn the tables, and instead consider fixing the permutation to some $w$ and varying the initial configuration $x$; we want to find a permutation that guarantees that we always obtain a maximal independent set $y \in \Kernels( G )$. As such, we call a permutation of $G$ that fixes $\Kernel(G)$ a \Define{permis} for $G$.

We now investigate which graphs have a permis. We first exhibit large classes of graphs that do have a permis in Theorem \ref{theorem:good_graphs}. For that purpose, we need to review some graph theory first.

A graph is a \Define{comparability graph} if there exists a partial order $\sqsubseteq$ on $V$ such that $uv \in E$ if and only if $u \sqsubset v$. The following are comparability graphs: complete graphs, bipartite graphs, permutation graphs, and interval graphs.

A vertex is \Define{simplicial} if its neighbourhood is a clique, i.e. if $\Neighbourhood[ s ] \subseteq \Neighbourhood[ v ]$ for all $v \in \Neighbourhood[ s ]$.

We now introduce an operation on graphs, that we call \Define{graph composition}. Let $H$ be an $n$-vertex graph, $G_1, \dots, G_n$ other graphs, then the composition $H(G_1, \dots, G_n)$ is obtained by replacing each vertex $v$ of $H$ by the graph $G_v$, and whenever $uv \in E(H)$, adding all edges between $G_u$ and $G_v$. This construction includes for instance the disjoint union of two graphs: $G_1 \cup G_2 = \bar{K}_2(G_1, G_2)$; the full union with all edges between $G_1$ and $G_2$: $K_2(G_1, G_2)$; adding an open twin (a new vertex $v'$ with $\Neighbourhood( v' ) = \Neighbourhood( v )$ for some vertex $v$ of $H$): $H(K_1, \dots, K_1, \bar{K_2}, K_1, \dots, K_1)$; similarly, adding a closed twin ($\Neighbourhood[ v' ] = \Neighbourhood[ v ]$).


\begin{theorem} \label{theorem:good_graphs}
Let $G$ be a graph. If $G$ satisfies any of the following properties, then $G$ has a permis:
\begin{enumerate}
    \item \label{item:seven_vertices}
    $G$ has at most seven vertices, and is not the heptagon $C_7$;
    
    \item \label{item:comparability}
    $G$ is a comparability graph; 

	\item \label{item:simplicial}
    the set of simplicial vertices of $G$ is a dominating set;

	\item \label{item:composition}
    $G = H(G_1, \dots, G_n)$, where each of $H, G_1, \dots, G_n$ has a permis.
\end{enumerate}
\end{theorem}

\begin{proof}[Sketch]
The proof of \ref{item:seven_vertices} is by computer search. For \ref{item:comparability}, the permis goes through the vertices ``from lowest to highest'' according to $\sqsubseteq$. For \ref{item:simplicial}, $G$ has a maximal independent set $M$ of simplicial vertices, and the permis visits $M$ last. For \ref{item:composition}, we can reduce ourselves to the case where only one vertex $b$ is blown up into a graph $G_b$. Then the permis for $G$ is obtained by taking the permis for $H$ and replacing the update of $b$ by a permis for $G_b$. Everything works as though the other vertices see $\bigjoin_{v \in G_b} x_v$. The full proof is given in Appendix \ref{appendix:proof_of_theorem:good_graphs}.



\end{proof}










We now exhibit classes of graphs without a permis. As mentioned in Theorem \ref{theorem:good_graphs}, the smallest graph without a permis is the heptagon.

\begin{proposition} \label{proposition:bad_cycles}
For all $2k+1 \ge 7$, the odd hole $C_{2k+1}$ does not have a permis.
\end{proposition}

\begin{proof}
Let $w$ be a permutation, and orient the edges such that $a \to b$ if and only if $a = w_i$, $b = w_j$ with $j > i$. We shall prove that there cannot be two consecutive arcs in the same direction; this shows that the direction of arcs must alternate, which is impossible because there is an odd number of arcs in the cycle. We do this by a case analysis on the arcs preceding those two consecutive arcs.

We consider six vertices $f, e, d, c, b, a$, where the last two arcs $c \to b \to a$ are in the same direction. The first case is where $d \to c$. In that case, if $(x_a, x_b, x_c) = (1,1,1)$, then $(y_b, y_c, y_d) = (0,0,0)$ as shown in Case 1 below along with the other three cases.

\begin{tikzpicture}
\begin{scope}[yshift=0, yscale=0.5]
    \node (x1) at (-1,2) {$x$};
    \node (xc1) at (3,2) {1};
    \node (xb1) at (4,2) {1};
    \node (xa1) at (5,2) {1};

    \node (C1) at (-2,2) {Case 1};
    \node (d1) at (2,1) {$d$};
    \node (c1) at (3,1) {$c$};
    \node (b1) at (4,1) {$b$};
    \node (a1) at (5,1) {$a$};

    \draw[-latex] (d1) -- (c1);
    \draw[-latex] (c1) -- (b1);
    \draw[-latex] (b1) -- (a1);

    \node (y1) at (-1,0) {$y$};
    \node (yd1) at (2,0) {0};
    \node (yc1) at (3,0) {0};
    \node (yb1) at (4,0) {0};
\end{scope}

\begin{scope}[yshift=-1.7cm, yscale=0.5]
    \node (x2) at (-1,2) {$x$};
    \node (xe2) at (1,2) {1};
    \node (xd2) at (2,2) {1};
    \node (xb2) at (4,2) {1};
    \node (xa2) at (5,2) {1};

    \node (C2) at (-2,2) {Case 2};
    \node (e2) at (1,1) {$e$};
    \node (d2) at (2,1) {$d$};
    \node (c2) at (3,1) {$c$};
    \node (b2) at (4,1) {$b$};
    \node (a2) at (5,1) {$a$};

    \draw[-latex] (d2) -- (e2);
    \draw[-latex] (c2) -- (d2);
    \draw[-latex] (c2) -- (b2);
    \draw[-latex] (b2) -- (a2);

    \node (y2) at (-1,0) {$y$};
    \node (yd2) at (2,0) {0};
    \node (yc2) at (3,0) {0};
    \node (yb2) at (4,0) {0};
\end{scope}

\begin{scope}[yshift=-3.4cm, yscale=0.5]
    \node (x3) at (-1,2) {$x$};
    \node (xe3) at (1,2) {1};
    \node (xd3) at (2,2) {0};
    \node (xb3) at (4,2) {1};
    \node (xa3) at (5,2) {1};

    \node (C3) at (-2,2) {Case 3};
    \node (f3) at (0,1) {$f$};
    \node (e3) at (1,1) {$e$};
    \node (d3) at (2,1) {$d$};
    \node (c3) at (3,1) {$c$};
    \node (b3) at (4,1) {$b$};
    \node (a3) at (5,1) {$a$};

    \draw[-latex] (f3) -- (e3);
    \draw[-latex] (e3) -- (d3);
    \draw[-latex] (c3) -- (d3);
    \draw[-latex] (c3) -- (b3);
    \draw[-latex] (b3) -- (a3);

    \node (y3) at (-1,0) {$y$};
    \node (yf3) at (0,0) {0};
    \node (ye3) at (1,0) {1};
    \node (yd3) at (2,0) {0};
    \node (yc3) at (3,0) {0};
    \node (yb3) at (4,0) {0};
\end{scope}

\begin{scope}[yshift=-5.1cm, yscale=0.5]
    \node (x4) at (-1,2) {$x$};
    \node (xf4) at (0,2) {0};
    \node (xd4) at (2,2) {0};
    \node (xb4) at (4,2) {1};
    \node (xa4) at (5,2) {1};

    \node (C4) at (-2,2) {Case 4};
    \node (f4) at (0,1) {$f$};
    \node (e4) at (1,1) {$e$};
    \node (d4) at (2,1) {$d$};
    \node (c4) at (3,1) {$c$};
    \node (b4) at (4,1) {$b$};
    \node (a4) at (5,1) {$a$};

    \draw[-latex] (e4) -- (f4);
    \draw[-latex] (e4) -- (d4);
    \draw[-latex] (c4) -- (d4);
    \draw[-latex] (c4) -- (b4);
    \draw[-latex] (b4) -- (a4);

    \node (y4) at (-1,0) {$y$};
    \node (ye4) at (1,0) {1};
    \node (yd4) at (2,0) {0};
    \node (yc4) at (3,0) {0};
    \node (yb4) at (4,0) {0};
\end{scope}

\end{tikzpicture}

\qed \end{proof}


Say a set of vertices $S$ is \Define{tethered} if there is an  edge $st$ between any $s \in S$ and and any $t \in T = \Neighbourhood(S) \setminus S$.

\begin{proposition} \label{proposition:blocking_set}
Let $G$ be a graph. If $G$ has a tethered set of vertices $S$ such that $G[ S ]$ has no permis, then $G$ has no permis.
\end{proposition}

\begin{proof}[Sketch]
Let $x$ be a configuration such that $x_S \ne 0$ and $x_T = 0$. Then updating $T$ will have no effect, and we have $y^a_T = 0$ throughout (for every $0\leq a \leq l$). Therefore, the updates in $S$ are the same as the updates in $G[S]$: $\Kernel^w( x; G )_S = \Kernel^{ \hat{w} }( x_S; G )$, where $\hat{w}$ represents the updates of $S$ only. Since $\hat{w}$ does not fix $G[S]$, $w$ does not fix $G$. The full proof is Appendix \ref{appendix:proof_of_proposition:blocking_set}.
\end{proof}

Propositions \ref{proposition:bad_cycles} and \ref{proposition:blocking_set} yield perhaps the second simplest class of graphs without a permis. The \Define{wheel graph} is $W_{n+1} = K_2( C_n, K_1 )$.

\begin{corollary} \label{corollary:bad_wheels}
For all $2k+2 \ge 8$, the wheel graph $W_{2k+2}$ does not have a permis.
\end{corollary}

An interesting consequence of Proposition \ref{proposition:blocking_set} is that having a permis is not a graph property that can be tested by focusing on an induced subgraph, even if the latter has all but seven vertices. Indeed, for any graph $H$, the graph $G = K_2( C_7, H )$ does not have permis, since the heptagon is tethered in $G$. Conversely, for any graph $H$ without a permis, adding a pending vertex $v'$ to each vertex $v$ of $H$ yields a graph $G$ where the set of simplicial vertices is a dominating set (all the vertices $v'$ form a maximal independent set of simplicial vertices). Therefore, some graphs with an induced heptagon do have a permis.

In general, the characterisation of graphs with a permis remains open.

\begin{question} \label{question:permis}
What is the complexity of the following decision problem: given $G$, does $G$ have a permis?
\end{question}





\appendix

\section{Proof of Theorem \ref{theorem:complexity_suffix}} \label{appendix:proof_of_theorem:complexity_suffix}

\begin{proof}
We prove that the \textsc{Dominion} problem is NP-complete. It is in NP: the certificate is the pair $(v, I)$ where $W \cap \Neighbourhood(v) \subseteq \Neighbourhood(I)$.

We show NP-hardness by first reducing  \textsc{Set Cover} to \textsc{Colony} and then reducing \textsc{Colony} to \textsc{Dominion}.

\begin{theorem} \label{theorem:complexity_colony}
    \textsc{Colony} is NP-complete.
\end{theorem}

\begin{proof}
    The proof is by reduction from \textsc{Set Cover}, which is NP-complete. 
    In \textsc{Set Cover}, the input is a finite set of elements $X=\{x_1,\ldots,x_n\}$, a collection  $C=\{C_1,C_2,\ldots,C_m\}$ of subsets of $X$, and an integer $k$. The question is whether there exists a subset $S\subseteq C$ of cardinality at most $k$ such that $\cup_{C_i\in S}C_i = X$. 
    
    We first construct the graph $G$ on $n+mk$ vertices. $G$ consists of: vertices $Q_j=\{q_j^1,\ldots, q_j^k\}$, for each $j\in [m]$; vertices $v_i$ for each $i\in [n]$; edges from each vertex in $Q_j$ to $v_i$, whenever $x_i\in C_j$; edges connecting $\{q_1^l,q_2^l,\ldots, q_m^l\}$ in a clique, for each $l\in [k]$. Let the target set $T=\{v_1,\ldots,v_n\}$. This concludes our construction; an illustrative example is shown in Fig. \ref{fig:setcover}.

    We now show that if $(X,C,k)$ is a yes-instance of \textsc{Set Cover}, then $(G,T)$ is a yes-instance of \textsc{Colony}.
    Let $S\subseteq C$  be a set cover of $X$ of cardinality at most $k$. 
    We obtain the set $I$ as follows: 
    $$I=\{q_j^a:C_j \text{ is the $a$th element of $S$}\}.$$
    Note that every node in $I$ exists in $G$ since $S$ has cardinality at most $k$ (the last subset to appear in $S$ is its $k$th element exactly). Further, $I$ is an independent set, since by construction every node $q_j^a$ is adjacent to some other node $q_l^b$ if and only if $a=b$. 
    Lastly, every node $v_i\in S$ is incident to some node in $I$; for any $i$, $\exists j: v_i\in C_j$. Then necessarily $\exists a: q_j^a\in I$, and by construction $(v_i, q_j^a)$ is an edge in $G$. 

    Conversely, if $(G,T)$ is a yes-instance of \textsc{Colony} then $(X,C,k)$ is a yes-instance of \textsc{Set Cover}.
    Let $I$ be an independent set in $G$ which colonizes $T$. By construction of $G$, $I$ has cardinality at most $k$. Suppose otherwise, for contradiction - then by the pigeon-hole principle there is some clique $C_j$ such that $|C_j\cap I|\geq 2$, contradicting that $I$ is an independent set.
    We obtain the set $S$ of cardinality $|I|$ as follows:
    $$S=\{C_j:\exists a \text{ such that }q_j^a\in I\}.$$
    We now show $S$ is a set cover of $X$. For each $i\in[n]$, $v_i$ must be adjacent to some node in $I$; denote this node $q_j^a$ - now by construction $x_i$ is in the set $C_j$, and $C_j \in S$. 

\qed \end{proof}

\begin{figure}
    \centering
    \includegraphics[page=1,width=\textwidth]{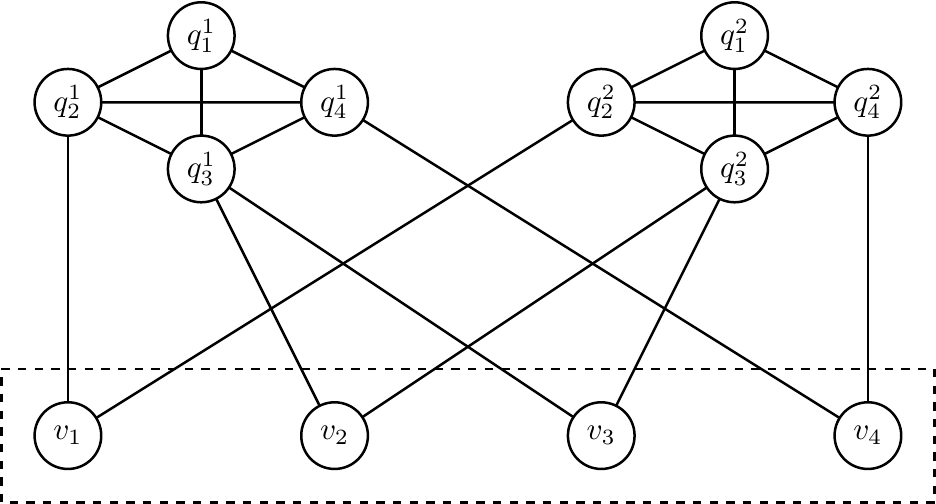}
    \caption{Illustration of the reduction from \textsc{Set Cover} to \textsc{Colony} (the set $T$ is the nodes in the dashed box). Here the \textsc{Set Cover} instance has $C_1=\emptyset,C_2=\{x_1\},C_3=\{x_2,x_3\},C_4=\{x_4\}$, with $k=2$. Observe that both the \textsc{Set Cover} instance and the \textsc{Colony} instance are no-instances.}
    \label{fig:setcover}
\end{figure}

\begin{theorem} \label{theorem:complexity_dominion}
\textsc{Dominion} is NP-complete.
\end{theorem} 

\begin{proof}
The proof is by reduction from \textsc{Colony}, which is NP-complete, as proved in Theorem \ref{theorem:complexity_colony}. Let $(G,S)$ be an instance of \textsc{Colony}, and construct the instance $(\hat{G}, \hat{S})$ as follows.

Let $G = (V,E)$ and denote $T = V \setminus S$. Then consider a copy $T' = \{t' : t \in T\}$ of $T$ and an additional vertex $\hat{v} \notin V \cup T'$. Let $\hat{G} = (\hat{V}, \hat{E})$ with $\hat{V} = V \cup T' \cup \{ \hat{v} \}$ and $\hat{E} = E \cup \{ tt' : t \in T \} \cup \{ s\hat{v} : s \in S \}$, and $\hat{S} = S \cup T'$. This construction is illustrated in Fig. \ref{fig:dominion}.

We only need to prove that $S$ is a colony of $G$ if and only if $\hat{S}$ is a colony of $\hat{G}$. Firstly, if $S$ is a colony of $G$, then there exists an independent set $I$ of $G$ such that $S \subseteq \Neighbourhood( I; G )$. Then $\hat{S} \cap \Neighbourhood( \hat{v}; \hat{G} ) = S$ is contained in $\Neighbourhood( I; \hat{G} - \hat{v} )$, thus $\hat{S}$ is a dominion of $\hat{G}$.

Conversely, if $\hat{S}$ is a dominion of $\hat{G}$, then there exists $u \in \hat{V} \setminus \hat{S}$ such that $\hat{S} \cap \Neighbourhood( u; \hat{G} )$ is a colony of $\hat{G} - u$. Then either $u = \hat{v}$ or $u \in T$. Suppose $u = t \in T$, then $t' \in \hat{S}$ is an isolated vertex of $G - t$, hence $\hat{S} \cap \Neighbourhood( t; \hat{G} )$ is not a colony of $\hat{G} - t$. Therefore, $u = \hat{v}$ and there exists an independent set $\hat{I}$ of $\hat{G} - \hat{v}$ such that $\hat{S} \cap \Neighbourhood( \hat{v} ; \hat{G} ) = S$ is contained in $\Neighbourhood( \hat{I}; \hat{G} )$. Since $S \subseteq V$ and $\Neighbourhood(S; \hat{G} - \hat{v}) \subseteq V$, we obtain $S \subseteq \Neighbourhood( \hat{I} \cap V; \hat{G} - \hat{v} ) \cap V = \Neighbourhood( \hat{I} \cap V; G )$, where $I = \hat{I} \cap V$ is an independent set of $G$. Thus, $S$ is a colony of $G$.
\qed \end{proof}

\begin{figure}
    \centering
    \includegraphics[page=2,width=.5\textwidth]{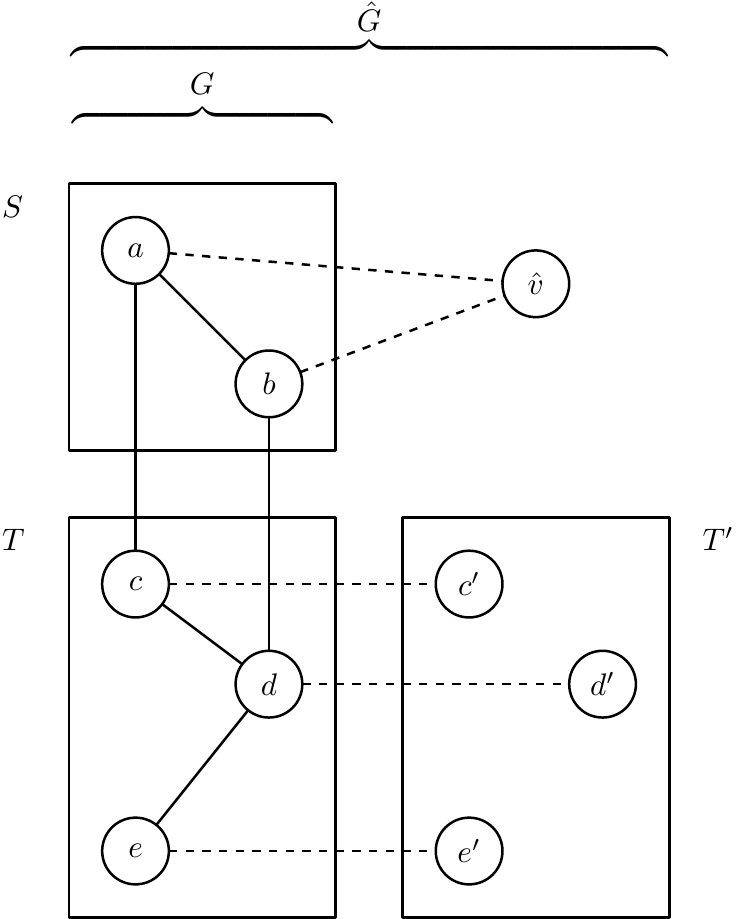}
    \caption{Example reduction from a no-instance of \textsc{Colony} $(G,S)$ to the corresponding no-instance of \textsc{Dominion} $(\hat{G}, \hat{S})$, with $\hat{S}:=S\cup T'$. 
}
    \label{fig:dominion}
\end{figure}

\qed \end{proof}

\section{Proof of Theorem \ref{theorem:complexity_fixing_set}} \label{appendix:proof_of_theorem:complexity_fixing_set}

\begin{proof}
The proof is by reduction from \textsc{Non-Dominion}, which is NP-hard, as proved in Theorem \ref{theorem:complexity_dominion}. Let $(G, S)$ be an instance of \textsc{Non-Dominion}, and construct the instance $( \hat{G}, \hat{S} )$ as follows.

Let $G = (V,E)$ and $T = V \setminus S$. For any $t \in T$, let $G_t = (V_t \cup \{ \hat{t} \}, E_t)$ be the graph defined as follows: $V_t = \{ u_t : u \in V \setminus t \}$  is a copy of all the vertices apart from $t$, which is replaced by a new vertex $\hat{t} \notin V_t$, and $E_t = \{ a_tb_t : ab \in E, a,b \ne t \} \cup \{ s_t \hat{t} : st \in E, s \in S \}$ is obtained by removing the edges between $t$ and the rest of $T$. Then $G$ is the disjoint union of all those graphs, i.e. $G = \bigcup_{t \in T} G_t$, while $\hat{S} = \bigcup_{t \in T} V_t$. For the sake of simplicity, we shall use the notation $A_t = \{ u_t : u \in A \}$ for all $A \subseteq V \setminus \{ t \}$.

By construction, $\hat{G} - \hat{S}$ is the empty graph on $\{ \hat{t} : t \in T \}$, hence $\hat{S}$ is a vertex cover of $\hat{G}$. All we need to show is that $\hat{S}$ is a non-dominion of $\hat{G}$ if and only if $S$ is a non-dominion of $G$. We have that $\hat{S}$ is a dominion of $\hat{G}$ if and only if there exists $\hat{t}$ and an independent set $\hat{I}$ of $\hat{G} - \hat{t}$ such that $W = \hat{S} \cap \Neighbourhood( \hat{t}; \hat{G} ) = ( S \cap \Neighbourhood( t; G ) )_t$ is contained in $\Neighbourhood( \hat{I} ; \hat{G} )$. We have $\hat{I} \cap V_t = I_t$ for some independent set $I$ of $G$. Since $W \subseteq V_t$ and $\Neighbourhood( W; \hat{G} - \hat{t} ) \subseteq V_t$, we have $W \subseteq \Neighbourhood( \hat{I} \cap V_t; \hat{G} ) \cap V_t = \Neighbourhood( I; G )_t$, which is equivalent to $S$ being a dominion of $G$.
\end{proof}

\begin{figure}
    \centering
    \includegraphics[page=3, width=\textwidth]{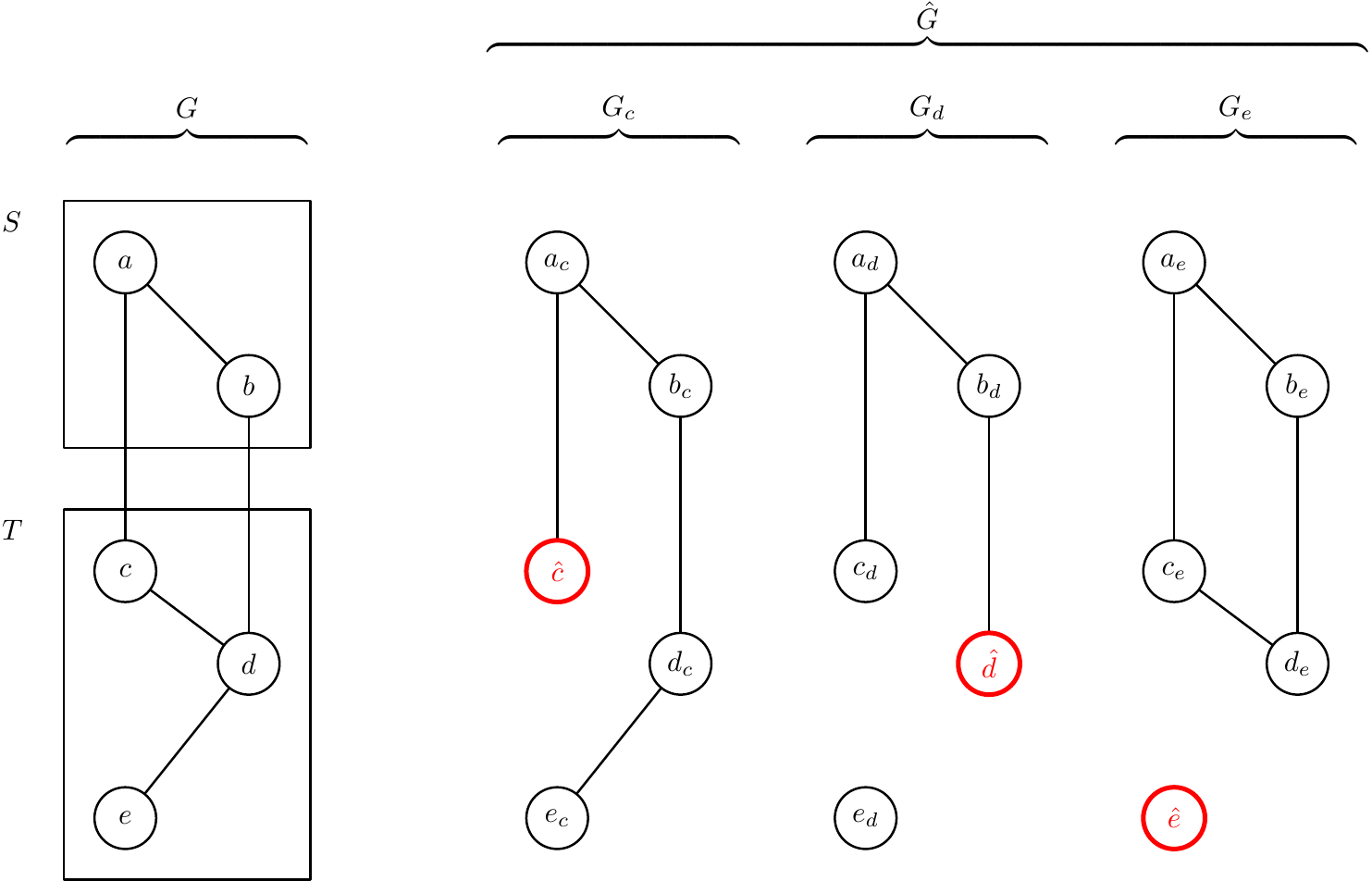}
    \caption{Example reduction from a no-instance of \textsc{Non-Dominion} $(G,S)$ to the corresponding no-instance of \textsc{Fixing Set} $(\hat{G}, \hat{S})$, with $\hat{S} = V_c \cup V_d \cup V_e$.}
    \label{fig:enter-label}
\end{figure}


\section{Proof of Theorem \ref{theorem:good_graphs}} \label{appendix:proof_of_theorem:good_graphs}

\begin{lemma} \label{lemma:simplicial_MIS}
Let $G$ be a graph such that its set $S$ of simplicial vertices is a dominating set. Then $S$ contains a maximal independent set.
\end{lemma}

\begin{proof}
Suppose for the sake of contradiction that $S$ does not contain a maximal independent set, i.e. that no dominating set contained in $S$ is independent. Let $T \subseteq S$ be a minimal dominating set, and let $t, t'$ be adjacent vertices of $T$ so that $\Neighbourhood[ t ] = \Neighbourhood[ t' ]$. Thus $T \setminus \{ t' \}$ is still a dominating set, which is the desired contradiction.
\qed \end{proof}

\begin{proof}[of Theorem \ref{theorem:good_graphs}]
\begin{enumerate}

    \item 
    Proof by computer search. Code available at \url{https://github.com/dave-ck/MISMax}

    \item 
    Let $G$ be a comparability graph, and order its vertices so that $v_i \sqsubseteq v_j \implies i \le j$. Then let $w = w_1 \dots w_n$ with $w_i = v_i$ for all $i \in [n]$. For the sake of contradiction, suppose that $y_{ \Neighbourhood[w_i] } = 0$ for some $i \in [n]$. Since $\Kernel( y^{i-1} )_{w_i} = y_{w_i} = 0$, we have $y^{i-1}_{ \Neighbourhood[ w_i ] } \ne 0$. Therefore, let $j = \max\{ k \in [n] : w_k  \sim w_i, y^{i-1}_{w_j} = 1 \}$; since $y^{i-1}_{w_k} = y_{w_k} = 0$ for all $k \le i-1$, we have $j \ge i + 1$. But $\Kernel( y^{j-1} )_{ w_j } = y_{w_j} = 0$, thus there exists $l = \max\{ k \in [n] : w_k \sim w_j, y^{j-1}_{w_j} = 1 \}$. Again, $l \ge j + 1$, and hence $y^{i-1}_{w_l} = 1$. However, $w_l \sim w_j$ and $w_j \sim w_i$ imply that $w_l \sim w_i$, thus $l \in \{ k \in [n] : w_k \sim w_i, y^{i-1}_{w_j} = 1 \}$ and $l \le j$, which is the desired contradiction.

    \item
    By Lemma \ref{lemma:simplicial_MIS}, let $M$ be a maximal independent set of simplicial vertices, and let $w$ be a permutation of $G$ such that the vertices of $M$ appear last: $w_1, \dots, w_{n - |M|} \notin M$ and $w_{n- |M| + 1}, \dots, w_n \in M$. Suppose for the sake of contradiction that $y_{\Neighbourhood[ v ]} = 0$ for some vertex $v$. Then there exists $m \in M$ such that $\Neighbourhood[ m ] \subseteq \Neighbourhood[ v ]$, thus $y_{\Neighbourhood[ m ]} = 0$. Suppose $m = w_{a+1}$, then $y^a_{ \Neighbourhood( w_{a+1} ) } = y_{ \Neighbourhood( m ) } = 0$, hence $y_m = \Kernel( y^a )_{ w_{a+1} } = 1$, which is the desired contradiction.
    
    \item
    It is easily shown that a graph composition can be obtained by repeatedly blowing up one vertex $b$ into the graph $G_b$ at a time. As such, we only need to prove the case where $G = H( K_1, \dots, K_1, G_b, K_1, \dots, K_1 )$, where the vertices are sorted according to a permis $\hat{w} = \hat{w}_1 \dots \hat{w}_b \dots \hat{w}_n$ of $H$. For any configuration $x$ of $G$, let $\hat{x}$ be the configuration of $H$ such that $\hat{x}_u = x_u$ for all $u \ne \hat{w}_b$ and $\hat{x}_{ \hat{w}_b} = \bigjoin_{v \in G_b} x_v$. Let $w^b$ be a permis of $G_b$ and consider the permutation $w$ of $G$ given by $w = \hat{w}_1 \dots \hat{w}_{b-1} w^b \hat{w}_{b+1} \dots \hat{w}_n$. We then prove that $y \in \Kernels( G )$ by considering the three main steps of $w$. We denote the vertex set of $G_b$ as $V_b$.
    \begin{itemize}
        \item Step 1: before the update of $G_b$. It is easy to show that for any $1 \le a < b$, we have $\Kernel^{w_1 \dots w_a}( x; G )_{ G - V_b } = \Kernel^{ \hat{w}_1 \dots \hat{w}_a }( \hat{x}; H )_{H - \hat{w}_b}$.

        \item Step 2: update of $G_b$. Note that $V_b$ is a tethered set of $G$, so let $T = \Neighbourhood( V_b; G ) \setminus V_b$. Let $\alpha = y^{b-1}$ be the initial configuration and $\beta = y^{b-1 + |V_b|}$ be the final configuration of the update of $G_b$. If $\alpha_{ T } \ne 0$, then the whole of $G_b$ will be updated to $0$: $\beta_{V_b} = 0$. Otherwise, it is as if $G_b$ is isolated from the rest of the graph and $\beta_{V_b} = \Kernel^{w^b}( \alpha_{V_b} ; G_b )$. In either case, we have $\hat{ \beta } = \Kernel^{ ( \hat{w}_b ) }( \hat{ \alpha }; H )$.

        \item Step 3: after the update of $G_b$. Again, we have for all $b < a \le n$, $\Kernel^{w_{b+1} \dots w_a}( \beta; G )_{ G - V_b } = \Kernel^{ \hat{w}_{b+1} \dots \hat{w}_a }( \hat{ \beta }; H )_{H - \hat{w}_b}$.
    \end{itemize}
    In conclusion, we have $y_{G - V_b} = \Kernel^{ \hat{w} }( \hat{x}; H )_{H - \hat{w}_b}$, and if $\Kernel^{ \hat{w} }( \hat{x} )_{ \hat{w}_b } = 0$ then $y_{ V_b } = 0$ else $y_{ V_b } = \Kernel^{w^b}( x_S; G_b )$. In either case, we obtain that $y \in \Kernels( G )$.
    
\end{enumerate}
\qed \end{proof}

\section{Proof of Proposition \ref{proposition:blocking_set}} \label{appendix:proof_of_proposition:blocking_set}

\begin{proof}
Let $w$ be a permutation of $G$ and $\hat{w}$ be the subsequence of $w$ satisfying $[\hat{w}]=S$. Let $\hat{x}$ be a configuration of $G[S]$ which is not fixed by $\hat{w}$: $\Kernel^{\hat{w}} ( \hat{x}; G[S] ) \notin \Kernels( G[S] )$. We first note that $\hat{x} \ne 0$ and that for all $0 \le a \le |\hat{w}|$, $\Kernel^{ \hat{w}_1 \dots \hat{w}_a }( \hat{x}; G[S] ) \ne 0$.

Let $T = \Neighbourhood( S ) \setminus S$ and $U = V \setminus ( S \cup T )$ and $x = ( x_S = \hat{x}, x_T = 0, x_U )$. We prove by induction on $0 \le b \le |w|$ that 
\[
    y^b := \Kernel^{ w_1 \dots w_b }( x; G ) = \left( y^b_S = \Kernel^{ \hat{w}_1 \dots \hat{w}_{ b' } }( \hat{x}; G[S] ) , y^b_T = 0, y^b_U \right),
\]
where $b'$ is defined by $[ \hat{w}_1 \dots \hat{w}_{ b' } ] = S \cap [ w_1 \dots w_b ]$. The base case $b = 0$ is clear. Suppose it holds for $b - 1$.
\begin{itemize}
    \item Case 1: $w_b \in S$. Then $b' = (b-1)' + 1$ and $w_b = \hat{w}_{ b' }$. Since $y^{b-1}_T = 0$, we have 
    \[
        y^b_{w_b} = \Kernel( y^{b-1} ; G )_{w_b} = \Kernel( y^{b-1}_S ; G[S] )_{w_b} = \Kernel(  \Kernel^{ \hat{w}_1 \dots \hat{w}_{ b' - 1 } }( \hat{x} ; G[S] ) ; G[S]  )_{ \hat{w}_{ b' } } = \Kernel^{ \hat{w}_1 \dots \hat{w}_{ b' } }( \hat{x}; G[S] )_{w_b},
    \]
    and hence $y^b_S = \Kernel^{ \hat{w}_1 \dots \hat{w}_{ b' } }( \hat{x}; G[S] )$.

    \item Case 2: $w_b \in T$. Then $b' = (b-1)'$. Since $y^{b-1}_S \ne 0$, we have $\Kernel( y^{b-1} ; G )_{w_b} = 0$ and hence $y^b_T = 0$.

    \item Case 3: $w_b \in U$. This case is trivial.
\end{itemize}

For $b = |w|$ we obtain $y = \Kernel^w(x ; G) = ( \Kernel^{ \hat{w} } ( \hat{x}; G[S] ), 0, y_U )$, for which $y_S \notin \Kernels( G[ S ] )$, and hence $y \notin \Kernels( G )$.
\qed \end{proof}

\end{document}